\documentclass[12pt,onecolumn]{IEEEtran}
\IEEEoverridecommandlockouts

\usepackage{graphicx}
\usepackage{psfrag}

\usepackage[cmex10]{amsmath}
\usepackage{amssymb}
\usepackage{amsthm}
\usepackage{amsfonts}

\usepackage{cases}
\usepackage{units}
\usepackage{dsfont}

\usepackage{cite}
\usepackage{url}

\newcommand{\E}{\mathrm{E}}

\providecommand{\abs}[1]{\lvert{#1}\rvert}
\providecommand{\asq}[1]{\abs{#1}^2}

\DeclareMathOperator{\dB}{dB}
\DeclareMathOperator{\Prob}{Pr}

\newtheorem{theorem}{Theorem}
\newtheorem{prop}{Proposition}
\newtheorem{lemma}{Lemma}

\begin{document}

\title{Completion Time Minimization and Robust Power Control in Wireless Packet Networks}
\author{Chris~T.~K.~Ng,
        Muriel~M\'{e}dard,
        Asuman~Ozdaglar%
\thanks{This material is based upon work under subcontract \#069145 and \#060786 issued by BAE Systems National Security Solutions, Inc.;
and supported by the Defense Advanced Research Projects Agency (DARPA) and the Space and Naval Warfare System Center (SPAWARSYSCEN), San Diego under Contract Nos.\ N66001-08-C-2013 and N66001-06-C-2020; and under ITMANET subcontract \#18870740-37362-C issued by Stanford University and supported by the DARPA\@.
The work of C.~Ng was supported by a Croucher Foundation Fellowship.
The material in this paper was presented in part at the IEEE International Conference on Communications, Dresden, Germany, June 2009.}%
\thanks{C.~Ng was with the Department of Electrical Engineering and Computer Science, Massachusetts Institute of Technology, Cambridge, MA 02139 USA, and is currently with Bell Labs, Alcatel-Lucent, Holmdel, NJ 07733 USA (e-mail: Chris.Ng@alcatel-lucent.com).}%
\thanks{M.~M\'{e}dard and A.~Ozdaglar are with the Department of Electrical Engineering and Computer Science, Massachusetts Institute of Technology, Cambridge, MA 02139 USA (e-mail: medard@mit.edu; asuman@mit.edu).}%
}

\maketitle
\thispagestyle{empty}

\begin{abstract}
A wireless packet network is considered in which each user transmits a stream of packets to its destination. The transmit power of each user interferes with the transmission of all other users. A convex cost function of the completion times of the user packets is minimized by optimally allocating the users' transmission power subject to their respective power constraints. At all ranges of SINR, completion time minimization can be formulated as a convex optimization problem and hence can be efficiently solved. In particular, although the feasible rate region of the wireless network is non-convex, its corresponding completion time region is shown to be convex. When channel knowledge is imperfect, robust power control is considered based on the channel fading distribution subject to outage probability constraints. The problem is shown to be convex when the fading distribution is log-concave in exponentiated channel power gains; e.g., when each user is under independent Rayleigh, Nakagami, or log-normal fading. Applying the optimization frameworks in a wireless cellular network, the average completion time is significantly reduced as compared to full power transmission.
\end{abstract}


\section{Introduction}
\label{sec:intro}

Power control plays a crucial role in the operation of a wireless network, as it strives to provide maximum benefits to the users within the confines of available resources.
In particular, in a wireless network, each user's transmit power interferes with the transmission of all other users; therefore, power allocation has a significant impact on the quality-of-service (QoS) experienced by the network users.
The user benefits derived from a given power allocation assignment can be characterized by different performance metrics.

Traditionally, in designing and evaluating the performance of a wireless network, a commonly used metric is a utility function of the user rates.
In particular, the network throughput can be characterized by maximizing the sum of the rates of the users.
However, for many applications, throughput is not the sole relevant performance metric.
In this paper, we study a different network performance metric that is motivated by packetized data applications.
Specifically, we consider the scenario where each user transmits a stream of packets to its destination, and we wish to minimize a convex cost function of the user packet completion times.
Each packet may represent a frame of a multimedia source, which the user wishes to receive as soon as possible.
We show that the minimization of the completion times can be formulated as a convex optimization problem,
and hence the corresponding optimal power allocation can be computed efficiently.

On the other hand, the use of the traditional throughput performance metrics present several challenges in the design of wireless data networks.
It has been recognized that maximizing a concave utility function over the feasible rate region in a wireless network is not necessarily a convex optimization problem \cite{chiang07:pow_ctrl_geo_prog, qian09:mapel_wl_pow_ctrl, boche04:cvx_fea_qos_region, boche03:log_cvx_mftp, oneill08:opt_adap_mod_wnum}.
In the high signal-to-interference-plus-noise (SINR) regime, rate utility maximization may be approximately formulated as a convex optimization problem \cite{chiang07:pow_ctrl_geo_prog}.
For traditional voice telephony applications that need to maintain at least a moderate minimum rate, the high-SINR regime is often an appropriate assumption.
However, in wireless sensor networks or low-power data applications, a user may wish to transmit at arbitrarily low rates (i.e., at low SINRs) depending on the channel conditions, and the high-SINR assumption may not always be applicable.

In this paper, we show that completion time minimization and the corresponding optimal power allocation can be formulated as a convex optimization problem at all ranges of SINR\@.
Moreover, we consider imperfect channel knowledge due to channel fading, and formulate the minimization as a stochastic programming problem.
The channel gains are modeled as random variables: an outage event occurs when the channel realization falls below the transmitter's SINR target.
Robust power control is considered where each user is subject to an outage probability constraint.
We show that for a wide class of commonly used channel fading distributions, e.g., Rayleigh, Nakagami, and log-normal, robust power control can be posed as a convex optimization problem.
We apply the optimization frameworks in the setting of a wireless cellular network, and show that optimizing transmission power can significantly reduce the average completion time as compared to full power transmission.

Optimal power control in wireless networks is studied in \cite{chiang07:pow_ctrl_geo_prog}; it shows that maximizing concave rate utility functions can be formulated as geometric programming (GP) problems in the high-SINR regime, which are convex and hence their solutions can be computed efficiently.
In the medium- to low-SINR regime, \cite{tan05:noncvx_pow_low_sir} describes an iterative approximation method to maximize a concave rate utility function by solving a series of GPs.
For weighted throughput maximization,
\cite{qian09:mapel_wl_pow_ctrl} proposes an algorithm to globally maximize a linear function of the rates by bounding the feasible SINR region by a series of polyblocks.
In \cite{boche04:cvx_fea_qos_region}, sufficient conditions are presented for the convexity of the feasible QoS region, with optional constraints on the allocation of user power.
The log-convexity of the SINR feasible region is characterized in \cite{sung02:log_cvx_sir_regn, boche03:log_cvx_mftp}.
Utility maximization through joint optimization of adaptive modulation, rate allocation, and power control is investigated in \cite{oneill08:opt_adap_mod_wnum}.
In \cite{yu02:dist_muser_pow_dsl}, power control is studied in frequency-selective Gaussian interference channels.
Outage probabilities corresponding to different fading distributions for network users and interferers are derived in \cite{yao92:coch_intf_cell}.
For interference-limited wireless networks, optimal power control is considered in \cite{kandukuri02:intf_fading_outage_prob} under Rayleigh fading subject to outage probability constraints.

The rest of the paper is organized as follows.
Section~\ref{sec:sys_mod} describes the channel model and the optimization framework.
In Section~\ref{sec:pkt_ctm}, minimizing a cost function of the completion times is posed as a convex optimization problem.
Section~\ref{sec:ro_pow_ctrl} considers robust power control against imperfect channel knowledge subject to outage probability constraints.
Numerical examples of minimizing completion times in a wireless cellular network are presented in Section~\ref{sec:wl_cell_net}.
Section~\ref{sec:conclu} concludes the paper.

\paragraph*{Notation}
In this paper, $\mathds{R}$ ($\mathds{R}_+, \mathds{R}_{++}$) is the set of real (nonnegative, positive) numbers, $\mathds{C}$ is the complex field,
and the dimensions of the corresponding vectors/matrices are indicated by superscripts.
$A^T$ is the transpose of a matrix $A$, $\mathbf{1}$ is a vector of $1$'s, $\E[\,\cdot\,]$ denotes expectation, and $\Prob\{\,\cdot\,\}$ denotes the probability of an event.

\section{System Model}
\label{sec:sys_mod}

\subsection{Wireless Channels}

Consider the scenario in which $M$ users are communicating in a wireless packet network.
Each User~$i$ consists of a Transmitter~$i$ and a corresponding Receiver~$i$, where $i=1,\dotsc,M$.
Transmitter~$i$ wishes to send a stream of equal-length packets to Receiver~$i$, where each packet has $L_i$ bits.
We assume a narrow-band complex additive white Gaussian noise (AWGN) channel model between the transmitters and receivers
\begin{align}
\label{eq:Yi_Gij_Xj_Zi}
Y_i &= \sum_{j=1}^M H_{ij} X_j + Z_i, \quad i = 1,\dotsc,M
\end{align}
where $Y_i\in\mathds{C}$ is the observed signal at Receiver~$i$, $X_j\in\mathds{C}$ is the signal sent by Transmitter~$j$, $H_{ij}\in\mathds{C}$ is the complex baseband channel from Transmitter~$j$ to Receiver~$i$, and $Z_i\in\mathds{C}$ is independent zero-mean circularly symmetric complex Gaussian (ZMCSCG) noise with power $N_i$.
In the subsequent sections, we consider different channel knowledge assumptions where $H_{ij}$'s may represent known constants or random variables.
Suppose Transmitter~$i$ has transmit power constraint $\bar{P}_i$.
When Transmitter~$i$ transmits at a power level of $P_i \leq \bar{P}_i$, in the capacity limit, the transmission rate $R_i$ achieved by User~$i$ is given by
\begin{align}
\label{eq:Ri_log_Si}
R_i &= B\log(1+S_i)
\end{align}
where $\log$ is base 2, $B$ is the channel bandwidth, and $S_i$ is the signal-to-interference-plus-noise ratio (SINR) at Receiver~$i$.
In this paper, interference cancellation schemes are not considered.
When the interference from the other transmitters is treated as noise, the SINR at Receiver~$i$ is
\begin{align}
\label{eq:Si_Hij_Pi_SINR}
S_i &= \frac{\asq{H_{ii}}P_i}{N_i + \sum_{j\neq i} \asq{H_{ij}}P_j}.
\end{align}
We consider a full buffer traffic model where each user has an infinite backlog of packets to be sent at the transmitter.
We assume $L_i$ is sufficiently large to allow transmission at near channel capacity using Gaussian signals.
The completion time of the transmission of each of User~$i$'s packet is given by
\begin{align}
\label{eq:Ti_Li_R_i}
T_i = L_i/R_i.
\end{align}
For example, an AWGN wireless packet network with $M=2$ users is illustrated in Fig.~\ref{fig:awgn_pkt_net_2}.

\begin{figure}
  \centering
  \psfrag{X1}[r][r]{\small $X_1$}
  \psfrag{X2}[r][r]{\small $X_2$}
  \psfrag{Y1}[l][l]{\small $Y_1$}
  \psfrag{Y2}[l][l]{\small $Y_2$}
  \psfrag{Z1}[][]{\small $Z_1$}
  \psfrag{Z2}[][]{\small $Z_2$}
  \psfrag{R1}[][]{\small $R_1$}
  \psfrag{R2}[][]{\small $R_2$}
  \psfrag{P1}[][]{\small $P_1$}
  \psfrag{P2}[][]{\small $P_2$}
  \psfrag{L1}[][]{\small $L_1$}
  \psfrag{L2}[][]{\small $L_2$}
  \psfrag{T1}[][]{\small $T_1$}
  \psfrag{T2}[][]{\small $T_2$}
  \psfrag{H11}[][]{\small $H_{11}$}
  \psfrag{H12}[][]{\small $H_{12}$}
  \psfrag{H21}[][]{\small $H_{21}$}
  \psfrag{H22}[][]{\small $H_{22}$}
  \includegraphics{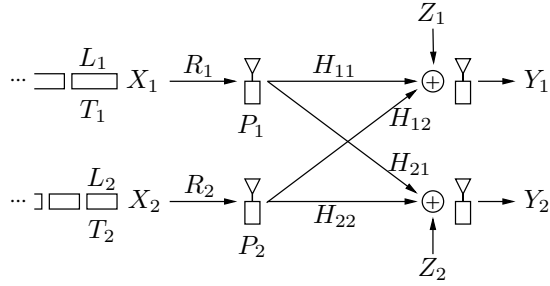}
  \caption{An AWGN wireless packet network with $M=2$ users.}
  \label{fig:awgn_pkt_net_2}
\end{figure}

\subsection{Completion Time Cost Function}

We consider the problem of minimizing a convex cost function of the completion times $T_1,\dots,T_M$, by optimally choosing the users' transmission power subject to the power constraints: $0 \leq P_i \leq \bar{P}_i$, $i=1,\dotsc,M$.
Let $J(\mathbf{T})$ be the completion time cost function, where
$\mathbf{T} \triangleq [T_1 \,\dots\, T_M]^T$,
and other vectors are denoted similarly in this paper.
We assume $J(\mathbf{T})$ is jointly-convex in $T_1,\dots,T_M$; the convexity penalizes overlong completion times.
For example, the following completion time cost functions are convex \cite{boyd04:convex_opt}:
\begin{align}
\label{eq:J_r_T}
J_r(\mathbf{T}) &= \mathbf{T}_{[1]} + \dotsb + \mathbf{T}_{[r]}\\
\label{eq:J_p_T}
J_p(\mathbf{T}) &= \bigl((T_1)^p + \dotsb + (T_M)^p\bigr)^{1/p}, \quad p \geq 1.
\end{align}
In (\ref{eq:J_r_T}), $\mathbf{T}_{[i]}$ denotes the $i$th largest component of $\mathbf{T}$.
Thus the cost function $J_r(\mathbf{T})$ is the sum of the $r$ longest completion times.
As special cases, $r=1$ represents the maximum completion time: $\max\{T_1,\dotsc,T_M\}$, and $r=M$ represents the sum of the completion times: $\sum_{i=1}^M T_i$.
In (\ref{eq:J_p_T}), the cost function $J_p(\mathbf{T})$ is the $\ell_p$-norm of the user completion times.
When $p$ is large, a more uniform distribution of the completion times will result in a lower cost.
Hence the parameter $p$ can be varied to achieve different fairness objectives with respect to resource allocation among the users.

\section{Packet Completion Time Minimization}
\label{sec:pkt_ctm}

\subsection{Perfect Channel Estimation}

We first consider the scenario where the channel gains $H_{ij}$'s can be accurately estimated and they are known by all users
\begin{align}
\asq{H_{ij}} &= G_{ij} \in \mathds{R}_+, \qquad i,j=1,\dotsc,M
\end{align}
where $G_{ij}$'s represent the channel power gains and they are known constants.
In this case, the minimization of the completion time cost function $J(\mathbf{T})$ can be mathematically formulated as the following optimization problem:
\begin{align}
\label{eq:opt_min_J_R}
&\text{minimize}\quad J(\mathbf{T})\\
&\text{over}\quad \mathbf{T}\in\mathds{R}_+^M,\; \mathbf{R}\in\mathds{R}_+^M,\; \mathbf{P}\in\mathds{R}_+^M\\
&\text{subject to}\notag\\
\label{eq:opt_Ti_Li_Ri}
& \quad T_i \geq L_i/R_i\\
\label{eq:opt_Ri_Gi_Pi}
& \quad R_i \leq B\log\biggl(1+\frac{G_{ii}P_i}{N_i + \sum_{j\neq i} G_{ij}P_j}\biggr)\\
\label{eq:opt_Pi_Pbi}
& \quad P_i \leq \bar{P}_i
\end{align}
where $i=1,\dotsc,M$, and the problem data $G_{ij}\in\mathds{R}_{+}$, $\bar{P}_i,N_i,L_i,B \in\mathds{R}_{++}$ are given.
The constraint (\ref{eq:opt_Ri_Gi_Pi}) in the optimization problem is not convex \cite{boche04:cvx_fea_qos_region, chiang07:pow_ctrl_geo_prog, oneill08:opt_adap_mod_wnum, qian09:mapel_wl_pow_ctrl}.
However, we show in the next section (\ref{eq:opt_min_J_R})--(\ref{eq:opt_Pi_Pbi}) can be transformed into a convex optimization problem, and hence its solution can be efficiently computed.

\subsection{Convex Optimization Formulation}
\label{sec:ctm_cvx_opt}

To formulate the completion time minimization problem given in (\ref{eq:opt_min_J_R})--(\ref{eq:opt_Pi_Pbi}) as a convex optimization problem, we first rewrite the constraints (\ref{eq:opt_Ti_Li_Ri}), (\ref{eq:opt_Ri_Gi_Pi}) as
\begin{align}
T_i &\geq \frac{L_i}{B\log(1+S_i)}, \quad i=1,\dotsc,M\\
\label{eq:Si_G_P_ij}
S_i &\leq \frac{G_{ii}P_i}{N_i + \sum_{j\neq i} G_{ij}P_j}, \quad i=1,\dotsc,M.
\end{align}
Next we apply the change of variables
\begin{align}
\label{eq:t_Si_t_Pi}
\tilde{S}_i &\triangleq \ln S_i,& \tilde{P}_i &\triangleq \ln P_i,& i &= 1,\dotsc,M
\end{align}
where $\ln$ is the natural logarithm.
The completion time minimization in (\ref{eq:opt_min_J_R})--(\ref{eq:opt_Pi_Pbi}) then becomes
\begin{align}
\label{eq:cvx_opt_min_J_R}
&\text{minimize}\quad J(\mathbf{T})\\
&\text{over}\quad \mathbf{T}\in\mathds{R}_+^M,\; \tilde{\mathbf{S}}\in\mathds{R}^M,\; \tilde{\mathbf{P}}\in\mathds{R}^M\\
&\text{subject to}\notag\\
\label{eq:cvx_opt_Ti_Li_xi}
& \quad T_i \geq \frac{L_i}{B\log\bigl(1+\exp(\tilde{S}_i)\bigr)}\\
\label{eq:cvx_opt_exp_xi_yi}
& \quad \tilde{S}_i -\tilde{P}_i +
\ln\biggl\{ N_i + \sum_{j\neq i} G_{ij} \exp(\tilde{P}_j)\biggr\} -\ln G_{ii} \leq 0\\
\label{eq:cvx_opt_yi_Pbi}
& \quad \tilde{P}_i \leq \ln \bar{P}_i
\end{align}
where $i=1,\dotsc,M$.
Note that the SINR constraint in (\ref{eq:cvx_opt_exp_xi_yi}) follows from rewriting (\ref{eq:Si_G_P_ij}) as
\begin{align}
S_i P_i^{-1} G_{ii}^{-1} N_i + \sum_{j \neq i} S_i P_i^{-1} P_j G_{ii}^{-1} G_{ij} \leq 1
\end{align}
and taking logarithm on both sides after applying (\ref{eq:t_Si_t_Pi}).
The change of variables is similar to the transformation techniques in geometric programming (GP) problems \cite{chiang07:pow_ctrl_geo_prog}.
In particular, the log-sum-exp function in constraint (\ref{eq:cvx_opt_exp_xi_yi}) is convex \cite{boyd04:convex_opt}.
The convexity of (\ref{eq:cvx_opt_Ti_Li_xi}) can be verified from its second-order conditions.
Specifically, the right-hand side of (\ref{eq:cvx_opt_Ti_Li_xi}) is twice-differentiable, and its second derivative is positive
\begin{align}
\frac{d^2}{dx^2}\bigl(\log(1+e^x)\bigr)^{-1} &=
\frac{e^x \ln 2 \bigl(2e^x - \ln(1+e^x)\bigr)}{(1+e^x)^2\bigl(\ln(1+e^x)\bigr)^3}\\
&> 0
\end{align}
which follows from the inequality $y>\ln(1+y)$ for $y>0$.
Note that the transformation in (\ref{eq:t_Si_t_Pi}) does impose a slight loss of generality as we assume $P_i \neq 0$.
Nevertheless, the formulation in (\ref{eq:cvx_opt_min_J_R})--(\ref{eq:cvx_opt_yi_Pbi}) is otherwise valid for all ranges of SINR, and its solution can be efficiently computed using standard numerical techniques in convex optimization, e.g., by the interior point method
\cite{renegar01:math_ipm_cvxopt, boyd04:convex_opt}.
Note that we may consider additional linear or convex constraints on $\mathbf{T}$, and sum power constraints on subsets of $\mathbf{P}$:
they can be readily incorporated in the optimization problem without violating its convexity.

\subsection{Fading Channels and Power Adaptation}
\label{sec:fadg_pow_adapt}

In this section, we consider fading channels, i.e., the channel gains $H_{ij}$'s in (\ref{eq:Yi_Gij_Xj_Zi}) experience random variations.
In particular, we assume the channel gains can be characterized by a set of $s\in\{1,\dotsc,S\}$ discrete fading states
\begin{align}
\asq{\mathbf{H}} &=
\begin{cases}
\mathbf{G}^{(1)} & \text{with probability $p_1$}\\
\quad\vdots & \\
\mathbf{G}^{(S)} & \text{with probability $p_S$},
\end{cases}&
\sum_{s=1}^{S} p_s &= 1,\quad p_s \geq 0
\end{align}
where
$\mathbf{H} \triangleq [H_{ij}] \in \mathds{C}^{M\times M}$
is the channel gain matrix,
the squared magnitude operation is taken component-wise,
and
$\mathbf{G}^{(s)} \triangleq [G_{ij}^{(s)}] \in \mathds{R}_+^{M\times M}$
are the known channel power gain realizations.
For example, the discrete states may represent a finite set of quantized channel estimates.
We consider slow fading where the duration of a fading state is long compared to the packet completion times.
We assume the channel state $s$ can be accurately estimated and it is known by all users,
i.e., all transmitters and receivers have perfect channel state information (CSI).
Power control when CSI is unavailable at the transmitters is treated in Section~\ref{sec:ro_pow_ctrl}.

We first consider the case where each user can adapt its transmission power level according to the fading state.
Suppose user~$i$ transmits at power level $P_i^{(s)}$ in fading state $s$, subject to the average power constraints
\begin{align}
\E[P_i] \triangleq \sum_{s=1}^S p_s P_i^{(s)} \leq \bar{P}_i,\quad i=1,\dotsc,M.
\end{align}
We wish to find the optimal power control policy $P_i^{(s)}$ with respect to the fading state $s$ for each User~$i$.
To minimize a cost function of the expected completion times, the optimization problem can be formulated as
\begin{align}
\label{eq:adapt_pow_min_J_ET}
&\text{minimize}\quad J(\E[\mathbf{T}])\\
&\text{over}\quad \E[\mathbf{T}]\in\mathds{R}_+^M,\; \mathbf{T}^{(s)}\in\mathds{R}_+^M,\;
\mathbf{S}^{(s)}\in\mathds{R}_+^M,\; \mathbf{P}^{(s)}\in\mathds{R}_+^M\\
&\text{subject to}\notag\\
\label{eq:adapt_pow_E_Ti}
& \quad \E[T_i] = \sum_{s=1}^S p_s T_i^{(s)}\\
\label{eq:adapt_pow_Ti_Li_Si}
& \quad T_i^{(s)} \geq \frac{L_i}{B\log\bigl(1+S_i^{(s)}\bigr)}\\
\label{eq:adapt_pow_Si_G_P_ij}
& \quad S_i^{(s)} \leq \frac{G_{ii}^{(s)}P_i^{(s)}}{N_i + \sum_{j\neq i} G_{ij}^{(s)}P_j^{(s)}}\\
\label{eq:adapt_pow_Pi_Pbi}
& \quad \sum_{s=1}^S p_s P_i^{(s)} \leq \bar{P}_i
\end{align}
where
$s=1,\dotsc,S$, $i=1,\dotsc,M$,
$\E[\mathbf{T}] \triangleq \bigl[\E[T_1] \,\dots\, \E[T_M]\bigr]^T$,
$\mathbf{T}^{(s)} \triangleq [T_1^{(s)} \,\dots\, T_M^{(s)}]^T$, and the vectors $\mathbf{S}^{(s)}$, $\mathbf{P}^{(s)}$ are defined similarly.
The optimization in (\ref{eq:adapt_pow_min_J_ET})--(\ref{eq:adapt_pow_Pi_Pbi}) can then be transformed into a convex optimization problem by similar techniques as described in Section~\ref{sec:ctm_cvx_opt}.
Note that to minimize the expected value of the cost function, it can be handled similarly by replacing the objective function in (\ref{eq:adapt_pow_min_J_ET}) by
\begin{align}
\label{eq:E_J_T}
\E[J(\mathbf{T})] = \sum_{s=1}^S p_s J(\mathbf{T}^{(s)})
\end{align}
where convexity is preserved in the nonnegative weighted sum of convex functions.

In the case where each user cannot adapt its transmission power level to the fading state (i.e., the transmitter is under a short-term power constraint),
the optimization problem is similar to (\ref{eq:adapt_pow_min_J_ET})--(\ref{eq:adapt_pow_Pi_Pbi}), but with the average power constraint in (\ref{eq:adapt_pow_Pi_Pbi}) replaced by separate power constraints for each fading state
\begin{align}
\label{eq:P_i_s_short_pow_constr}
P_i^{(s)} \leq \bar{P}_i, \quad s=1,\dotsc,S, \quad i=1,\dotsc,M.
\end{align}
Note that under the short-term power constraints of (\ref{eq:P_i_s_short_pow_constr}), minimizing the expected cost function (\ref{eq:E_J_T}) decomposes into $S$ independent optimization problems:
i.e., each of $J(\mathbf{T}^{(s)})$, for $s=1,\dotsc,S$, can be minimized separately.

\subsection{Relations to Rate Utility Maximization}
\label{sec:rate_region}

In general, in a wireless network as defined in (\ref{eq:Yi_Gij_Xj_Zi})--(\ref{eq:Ti_Li_R_i}), minimizing a convex cost function $J(\mathbf{T})$ of the completion times $T_1,\dotsc,T_M$ is not equivalent to maximizing a concave utility function $U(\mathbf{R})$ of the rates $R_1,\dotsc,R_M$.
In particular, maximizing $U(\mathbf{R})$ over the rate region is in general non-convex \cite{boche04:cvx_fea_qos_region}:
at high SINR it can be approximately formulated as a GP, and in the medium- to low-SINR regime there are iterative approximation methods \cite{chiang07:pow_ctrl_geo_prog}.
Suppose the cost function $J_+(\mathbf{T})$ is convex and nondecreasing in each argument $T_i$,
then completion time minimization is a special case of rate utility maximization where the optimization problem can be formulated as convex.
To see that minimizing $J_+(\mathbf{T})$ can be posed as a rate utility maximization problem, we define the corresponding rate utility function
\begin{align}
U_T(\mathbf{R}) &\triangleq -J_+(L_1/R_1,\dotsc,L_M/R_M).
\end{align}
Note that minimizing $J_+(\mathbf{T})$ is equivalent to maximizing $U_T(\mathbf{R})$, and the utility function $U_T(\mathbf{R})$ is concave in $\mathbf{R}$ as prescribed by the convexity composition rules \cite{boyd04:convex_opt}.
Therefore, in general, a rate utility maximization method can be used to minimize $J_+(\mathbf{T})$.
On the other hand, it is not true that a completion time minimization method is applicable in maximizing any general concave rate utility functions.
For example, consider a rate utility optimization problem: $\;\text{maximize}\;U(\mathbf{R})$,
where $U(\cdot)$ is concave in $\mathbf{R}$.
A naive approach may attempt to reformulate the above optimization into a completion time minimization problem as:
$\;\text{minimize}\;-U(L_1/R_1,\dotsc,L_M/R_M)$.
However,
since convexity is not preserved when a convex function is composed with inverses,
the resulting cost function $-U(\cdot)$ is not necessary convex in $\mathbf{R}$.

Nevertheless, in the converse, some rate utility maximization problems can be formulated as minimizing convex functions of the completion times.
Suppose we minimize a nonnegatively weighted sum of the completion times
\begin{align}
\label{eq:J_w_T}
J_w(\mathbf{T}) = a_1 T_1 + \dotsb + a_M T_M, \quad \mathbf{a}\triangleq[a_1 \,\dots\, a_M]^T\in\mathds{R}_+^M
\end{align}
then it is equivalent to maximizing
\begin{align}
U_d(\mathbf{R}) = -\frac{a_1'}{R_1} - \dotsb - \frac{a_M'}{R_M}
\end{align}
where $U_d(\mathbf{R})$ is the utility function that corresponds to minimum potential delay fairness \cite{srikant03:math_int_cong_ctrl}, with $a_i' \triangleq a_i L_i$, $i=1,\dotsc,M$.
In addition, minimizing $J_w(\mathbf{T})$ in (\ref{eq:J_w_T}) is also equivalent to maximizing the weighted harmonic mean of the rates
\begin{align}
U_h(\mathbf{R}) & = \biggl(\frac{a_1'}{R_1} + \dotsb + \frac{a_M'}{R_M}\biggr)^{-1}.
\end{align}
Note that by applying Jensen's inequality on the convex function $1/x$ for $x\in\mathds{R}_{++}$, we have
\begin{align}
\label{eq:inv_x_jensen}
\frac{1}{a_1'R_1 + \dotsb + a_M'R_M} \leq \frac{a_1'}{R_1} + \dotsb + \frac{a_M'}{R_M}
\end{align}
which implies
\begin{align}
\label{eq:w_R_Uh}
a_1'R_1 + \dotsb + a_M'R_M \geq U_h(\mathbf{R}).
\end{align}
Hence maximizing $U_h(\mathbf{R})$ provides a lower bound to $\max\, a_1'R_1 + \dotsb + a_M'R_M$, which represents a weighted throughput of the wireless network.
In particular, the bound is tight when $R_1 = \dotsb = R_M$, as equality is achieved in (\ref{eq:inv_x_jensen}).
Therefore, maximizing the minimum rate in (\ref{eq:U_n_R}) below can be formulated as a minimization of the convex cost function $J_x(\mathbf{T})$ as given in (\ref{eq:J_x_T}), which corresponds to the maximum completion time
\begin{align}
\label{eq:U_n_R}
U_n(\mathbf{R}) &= \min\{R_1,\dotsc,R_M\}\\
\label{eq:J_x_T}
J_x(\mathbf{T}) &= \max\{T_1,\dotsc,T_M\}.
\end{align}

Moreover, the entire rate region achievable under (\ref{eq:opt_Ri_Gi_Pi})--(\ref{eq:opt_Pi_Pbi}) can be characterized in terms of the corresponding completion time region.
Specifically, the completion time region as characterized in (\ref{eq:cvx_opt_Ti_Li_xi})--(\ref{eq:cvx_opt_yi_Pbi}) is convex,
and its boundary are given by the minimizer of $J_w(\mathbf{T})$ in (\ref{eq:J_w_T}) over all $\mathbf{1}^T \mathbf{a} = 1$, $\mathbf{a}\in\mathds{R}_+^M$.
In turn, from the monotonicity of (\ref{eq:Ti_Li_R_i}), each minimal completion time vector $\mathbf{T}^\star = \arg \min J_w(\mathbf{T})$ corresponds to a maximal rate vector $\mathbf{R}^\star$ on the boundary on the rate region (\ref{eq:opt_Ri_Gi_Pi})--(\ref{eq:opt_Pi_Pbi}), with $R_i^\star = L_i/T_i^\star$, for $i=1,\dotsc,M$.
As a numerical example, we consider the following 2-user AWGN wireless packet network:
\begin{align}
\mathbf{G} &= \begin{bmatrix}0.42 & 0.89\\ 0.63 & 0.15\end{bmatrix},
& \mathbf{L} &= \begin{bmatrix}100\\100\end{bmatrix}\\
\label{eq:nu_ex_Pb_N}
\bar{\mathbf{P}} &= \begin{bmatrix}0\\0\end{bmatrix} \dB,
& \mathbf{N} &= \begin{bmatrix}0\\0\end{bmatrix} \dB
\end{align}
with $B=0.1\,\mathrm{MHz}$, and maximum completion time constraints: $T_i \leq 100\,\mathrm{ms}$.
The completion time region and its corresponding rate region are shown in Fig.~\ref{fig:T1_T2} and Fig.~\ref{fig:R1_R2}, respectively.
Note that the power constraints as given in (\ref{eq:nu_ex_Pb_N}) belong to the low SINR regime, where the high-SINR GP approximation does not readily apply.
The completion time region is convex in Fig.~\ref{fig:T1_T2}; however, note that its rate region counterpart is non-convex as can be observed in Fig.~\ref{fig:R1_R2}.

\begin{figure}
  \centering
  \includegraphics*[width=12cm]{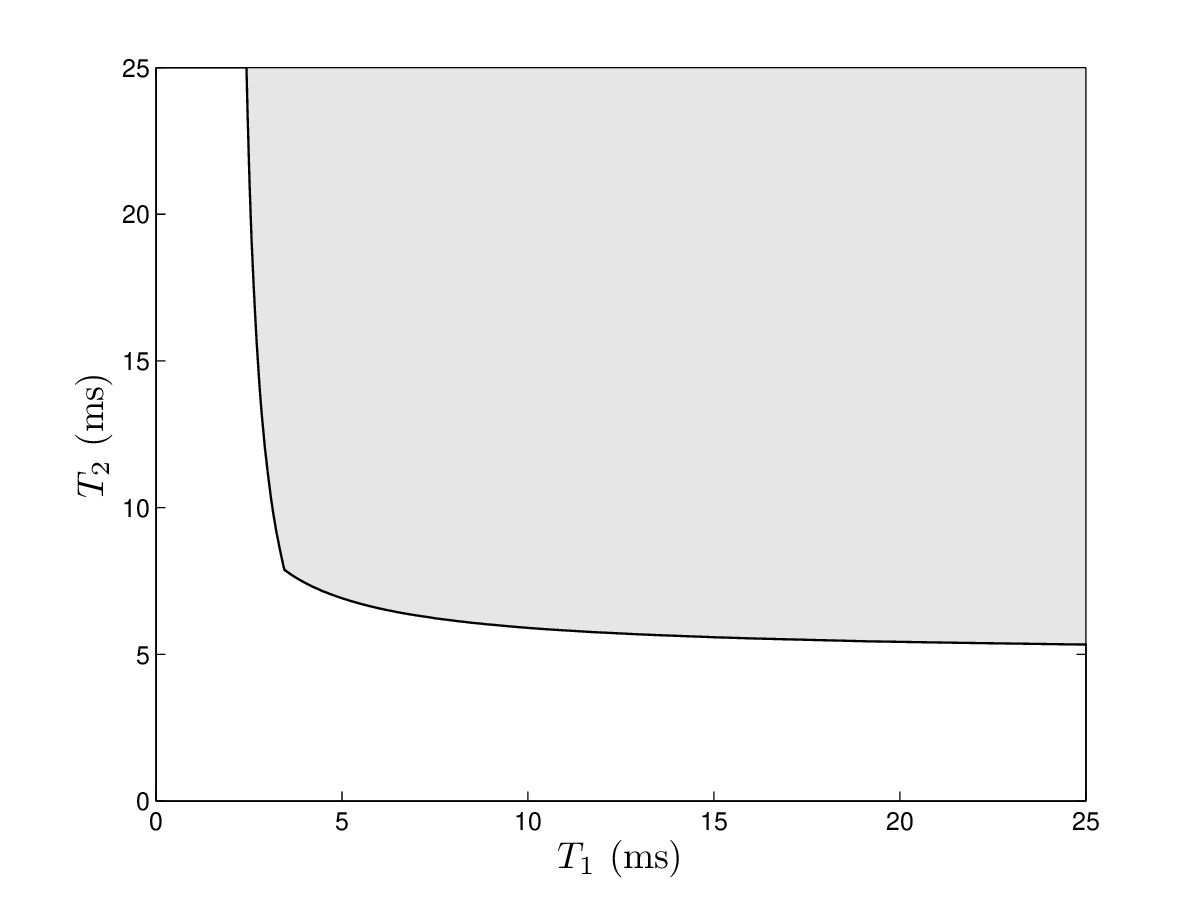}
  \caption{Completion time region ($\bar{P}_1$ = $\bar{P}_2$ = $0\,\dB$).}
  \label{fig:T1_T2}
\end{figure}

\begin{figure}
  \centering
  \includegraphics*[width=12cm]{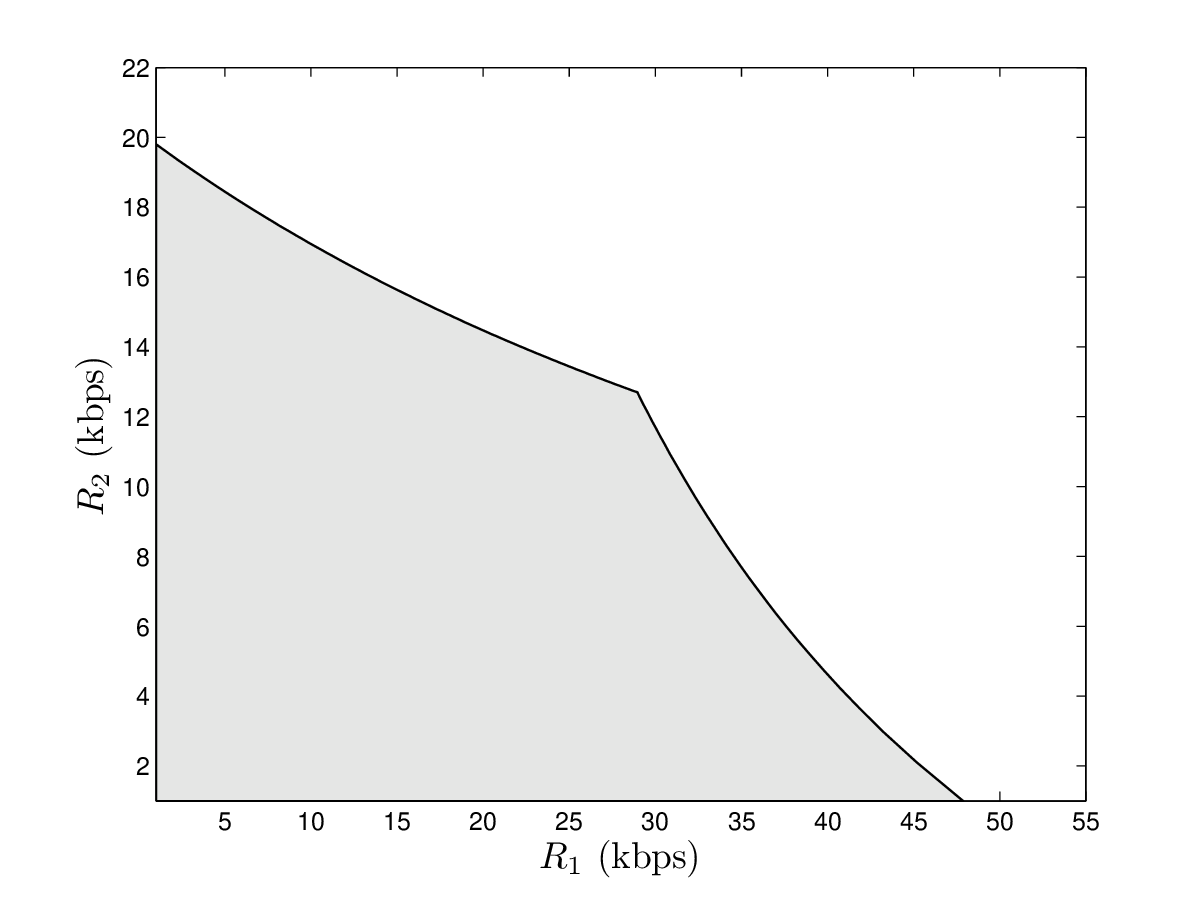}
  \caption{Rate region ($\bar{P}_1$ = $\bar{P}_2$ = $0\,\dB$).}
  \label{fig:R1_R2}
\end{figure}

\section{Robust Power Control}
\label{sec:ro_pow_ctrl}

\subsection{Outage Probability Constraints}

In Section~\ref{sec:pkt_ctm}, we assume that the channel gains $H_{ij}$'s can be accurately estimated.
However, in a fading environment where the channel estimates are updated not as fast as the channels vary, the transmitters may not know the $H_{ij}$'s perfectly.
In this section, we consider the scenario where the channel gains are modeled as random variables
\begin{align}
\asq{\mathbf{H}} &= \mathbf{W} \triangleq [W_{ij}] \in\mathds{R}_+^{M\times M}
\end{align}
where the transmitters know the joint probability distribution of $W_{ij}$'s but do not know their realization (the receivers have perfect channel knowledge).
Therefore, the transmitters have only channel distribution information (CDI) but not instantaneous CSI about the fading states.
As in Section~\ref{sec:fadg_pow_adapt}, we assume the duration of a fading state is long compared to the packet completion times.
Based on the channel distribution, each User~$i$ chooses a target SINR $S_i$.
Should the realized channel SINR fall below the target $S_i$, the receiver cannot decode the transmitter's message, and it results in an \emph{outage} event.
To ensure the network operates with an acceptable level of reliability, we extend the completion time minimization framework in Section~\ref{sec:pkt_ctm} to additionally consider constraints imposed on the permissible probability of outage.
Specifically, we minimize the completion time cost function subject to a set of outage probability constraints: $q_i \in[0,1]$, $i=1,\dotsc,M$,
where we stipulate that the probability of User~$i$'s transmission in outage not exceed $q_i$.

Incorporating the outage probability constraints $q_i$'s, the minimization of the completion time cost function is described by the following stochastic programming \cite{prekopa95:stoc_prog} problem:
\begin{align}
\label{eq:ro_op_min_J_R}
&\text{minimize}\quad J(\mathbf{T})\\
&\text{over}\quad \mathbf{T}\in\mathds{R}_+^M,\; \mathbf{S}\in\mathds{R}_+^M,\; \mathbf{P}\in\mathds{R}_+^M\\
&\text{subject to}\notag\\
\label{eq:ro_op_Ti_Li_Si}
& \quad T_i \geq \frac{L_i}{B\log(1+S_i)}\\
\label{eq:ro_op_Si_Gi_W_i_Pi}
& \quad \Pr\biggl\{\frac{W_{ii}P_i}{N_i + \sum_{j\neq i} W_{ij}P_j} \leq S_i \biggr\} < q_i\\
\label{eq:ro_op_Pi_Pbi}
& \quad P_i \leq \bar{P}_i
\end{align}
where $i=1,\dotsc,M$.
In the following,
we show that the minimization in (\ref{eq:ro_op_min_J_R})--(\ref{eq:ro_op_Pi_Pbi}) can be posed as a convex optimization problem for a wide class of channel fading distributions commonly considered in wireless communications.

\subsection{Reliability Functions}

In terms of the transformed variables in (\ref{eq:t_Si_t_Pi}), we first define the \emph{reliability function} as
\begin{align}
\Phi_i(\tilde{S}_i, \mathbf{\tilde{P}}) &\triangleq \Prob\{\text{User~$i$ not in outage}\}\\
\label{eq:Phi_i_Pr_W_S_P}
& = \Prob\biggl\{\tilde{W}_{ii} > \ln \Bigl\{N_i\exp(\tilde{S}_i-\tilde{P}_i)
+ \sum_{j\neq i} \exp(\tilde{S}_i-\tilde{P}_i+\tilde{P}_j+\tilde{W}_{ij})
\Bigr\}\biggr\}
\end{align}
where (\ref{eq:Phi_i_Pr_W_S_P}) follows from rearranging (\ref{eq:ro_op_Si_Gi_W_i_Pi}) with the additional change of variables
\begin{align}
\label{eq:t_W_ij_ln_W_ij}
\tilde{W}_{ij} &\triangleq \ln W_{ij}, \quad i,j=1,\dotsc,M.
\end{align}
Next, we characterize the reliability function in terms of the channel distribution.
Let $\mathbf{W}_i \in \mathds{R}_+^M$ denote the aggregate channel power gains from all transmitters to Receiver~$i$.
Therefore, $\mathbf{W}_i$ is an $M$-component nonnegative random vector that corresponds to the $i$th row of the channel gain matrix $\mathbf{W}$
\begin{align}
\mathbf{W}_i \triangleq \begin{bmatrix}W_{i1} &\dots &W_{iM}\end{bmatrix}^T.
\end{align}
Further, let $\mathbf{w}_i \in \mathds{R}_+^M$ be a realization of $\mathbf{W}_i$.
Under (\ref{eq:t_W_ij_ln_W_ij}), the transformed vectors $\tilde{\mathbf{W}}_i$, $\tilde{\mathbf{w}}_i$ are defined similarly.
Let $f_{\mathbf{W}_i}(\mathbf{w}_i)$ denote the joint probability distribution function (PDF) of $\mathbf{W}_i$.
Theorem~\ref{thm:rlb_fn_log_ccv} below describes the sufficient condition that establishes the log-concavity of $\Phi_i(\tilde{S}_i, \tilde{\mathbf{P}})$, under which (\ref{eq:ro_op_min_J_R})--(\ref{eq:ro_op_Pi_Pbi}) can be posed as the following convex optimization problem:
\begin{align}
\label{eq:ro_exp_min_J_R}
&\text{minimize}\quad J(\mathbf{T})\\
&\text{over}\quad \mathbf{T}\in\mathds{R}_+^M,\; \tilde{\mathbf{S}}\in\mathds{R}^M,\; \tilde{\mathbf{P}}\in\mathds{R}^M\\
&\text{subject to}\notag\\
\label{eq:ro_exp_Ti_Li_Si}
& \quad T_i \geq \frac{L_i}{B\log\bigl(1+\exp(\tilde{S}_i)\bigr)}\\
\label{eq:ro_exp_Si_Gi_W_i_Pi}
& \quad \ln \Phi_i(\tilde{S}_i, \tilde{\mathbf{P}}) \geq \ln(1-q_i)\\
\label{eq:ro_exp_Pi_Pbi}
& \quad \tilde{P}_i \leq \ln \bar{P}_i
\end{align}
where $i=1,\dotsc,M$.
In the following,
let $\exp(\mathbf{w}_i)$ denote component-wise exponentiation of the vector $\mathbf{w}_i$.
\begin{theorem}
\label{thm:rlb_fn_log_ccv}
The reliability function $\Phi_i(\tilde{S}_i, \tilde{\mathbf{P}})$ is log-concave in $\tilde{S}_i, \tilde{\mathbf{P}}$ if
$f_{\mathbf{W}_i}\bigl(\exp(\mathbf{w}_i)\bigr)$ is log-concave in $\mathbf{w}_i$.
\end{theorem}
The proof is given in Appendix~\ref{sec:prf_thm_rlb_fn_log_ccv}.
The following proposition shows that the condition given in Theorem~\ref{thm:rlb_fn_log_ccv}
is satisfied in a wide class of commonly used wireless channel fading distributions.
\begin{prop}
\label{prop:ray_naka_lognorm}
The condition given in Theorem~\ref{thm:rlb_fn_log_ccv} is satisfied
when each channel experiences independent fading distributed as:
i)~Rayleigh, ii)~Nakagami, or iii)~log-normal.
\end{prop}
The proof of Proposition~\ref{prop:ray_naka_lognorm} is given in Appendix~\ref{sec:prf_prop_ray_naka_lognorm}.
Rayleigh fading is commonly used to model richly scattered environments;
Nakagami models significant line-of-sight propagation (or is used to approximate the Rician fading distribution);
and the log-normal distribution is typically used to model the effects of shadowing due to signal attenuation through obstacles \cite{stuber00:prin_mob_comm}.
Therefore, in all these cases, completion time minimization subject to outage probability constraints can be formulated as convex optimization problem (\ref{eq:ro_exp_min_J_R})--(\ref{eq:ro_exp_Pi_Pbi}).

\subsection{Independent Rayleigh Fading}

As an example of the robust power control formulation,
let us consider the scenario in which each channel power gain $W_{ij}$ exhibits independent Rayleigh fading (with mean $G_{ij}$).
Thus $W_{ij}$ is distributed exponentially as
\begin{align}
f_{W_{ij}}(w_{ij}) = G_{ij}^{-1}\exp(-w_{ij}/G_{ij}),\quad w_{ij}\geq0, \quad i,j = 1,\dotsc,M
\end{align}
where $G_{ij}$ is a known constant that represents the average channel power gain.
In this case the reliability probability can be written as follows \cite{kandukuri02:intf_fading_outage_prob}:
\begin{align}
\Prob\{\text{User~$i$ not in outage}\} &= \exp\biggl(-\frac{S_iN_i}{G_{ii}P_i}\biggr) \prod_{j\neq i} \biggl(1+\frac{S_iG_{ij}P_j}{G_{ii}P_i}\biggr)^{-1}.
\end{align}
The logarithm of the reliability function then evaluates to
\begin{align}
\label{eq:Ray_ln_rel_fcn}
\ln \Phi_i(\tilde{S}_i,\tilde{\mathbf{P}})
&= -(N_i/G_{ii})\exp(\tilde{S}_i-\tilde{P}_i) - \sum_{j\neq i} \ln \Bigl\{1+(G_{ij}/G_{ii})\exp\bigl(\tilde{S}_i + \tilde{P}_j - \tilde{P}_i\bigr)\Bigr\}
\end{align}
which can be verified to be a concave function.
Therefore, under independent Rayleigh fading channels, the robust power control problem is
\begin{align}
\label{eq:Ray_min_J_R}
&\text{minimize}\quad J(\mathbf{T})\\
&\text{over}\quad \mathbf{T}\in\mathds{R}_+^M,\; \tilde{\mathbf{S}}\in\mathds{R}^M,\; \tilde{\mathbf{P}}\in\mathds{R}^M\\
&\text{subject to}\notag\\
\label{eq:Ray_Ti_Li_Si}
& \quad T_i \geq \frac{L_i}{B\log\bigl(1+\exp(\tilde{S}_i)\bigr)}\\
\label{eq:Ray_Si_Gi_W_i_Pi}
& \quad -(N_i/G_{ii})\exp(\tilde{S}_i-\tilde{P}_i) - \sum_{j\neq i} \ln \Bigl\{1+(G_{ij}/G_{ii})\exp\bigl(\tilde{S}_i + \tilde{P}_j - \tilde{P}_i\bigr)\Bigr\} +  \geq \ln(1-q_i)\\
\label{eq:Ray_Pi_Pbi}
& \quad \tilde{P}_i \leq \ln \bar{P}_i
\end{align}
where $i=1,\dotsc,M$, and the SINR constraint in (\ref{eq:Ray_Si_Gi_W_i_Pi}) follows from the reliability function under independent Rayleigh fading as given in (\ref{eq:Ray_ln_rel_fcn}).

\section{Wireless Cellular Networks}
\label{sec:wl_cell_net}

In this section, we consider the completion time minimization and robust power control frameworks developed in
Sections~\ref{sec:pkt_ctm} and~\ref{sec:ro_pow_ctrl},
and apply them in the setting of a wireless cellular network.
We assume the users do not cooperate in the network, and interference is treated as noise.
Power minimization in cellular networks subject to minimum rate constraints is studied in
\cite{foschini93:dist_auto_pow_conv, yates95:frwk_ul_pow_ctrl_cell}.
However, rate maximization subject to transmit power constraints remains an open problem.
In the following, we consider completion time as the performance metric, and present numerical examples where transmission power is optimized when every user has global channel knowledge, and in the case when the transmitters have only channel distribution information.

Let us consider a cellular network that consists of two rings of hexagonal cells, as illustrated in Fig.~\ref{fig:hex_cell}.
To minimize boundary effects, we assume wraparound at the edges of the network.
Each cell has three sectors; thus there are $19$ cells, or $57$ sectors, in the network, where each sector corresponds to one base station.
We consider the downlink channel: each base station has one transmit antenna and wishes to send information to one mobile, and each mobile has a single receive antenna.
Hence, there are $M=57$ users in the network.

We assume parameters that correspond to a typical urban outdoor cellular environment \cite{NGMN07:perf_eval_meth}.
The distance between any two closest cell centers is $0.5\,\mathrm{km}$.
Average channel SNR is determined by propagation path-loss (with a path-loss exponent of $3.76$) and log-normal shadow fading
(with $8$-$\dB$ standard deviation, $0.05$-$\mathrm{km}$ decorrelation distance, and $0.5$ correlation across base stations).
The transmit antenna at each sector has a parabolic beam pattern.
The mobiles are randomly populated in the network.
A mobile is associated with the base station to which it has the highest average SNR (up to the maximum of one mobile per base station).
The mobiles are indexed such that Base~$i$ wishes to transmit to Mobile~$i$.
In composition with the path-loss and shadowing, each channel also experiences i.i.d. fast Rayleigh fading.
Each base station is under a transmit power constraint, which corresponds to a cell-edge average SNR of $20\,\dB$.
We assume short-term power constraints where power allocation across fading states is not allowed.
In the network, each mobile suffers interference from all other base stations:
i.e., a frequency reuse factor of $1$ is assumed.
The wireless channel has bandwidth $0.1\,\mathrm{MHz}$,
and we assume packet length $L_i = 100$, normalized receiver noise power $N_i = 0\,\dB$, for all $i=1,\dotsc,57$.

\begin{figure}
  \centering
  \includegraphics[scale=1.5]{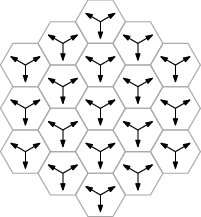}
  \caption{Hexagonal three-sectored cellular wireless network. There are $19$ cells in the network, with wraparound at the edges. Each cell has three sectors.
  Each sector corresponds to a base station, and each base station serves one mobile.
  Each arrow represents the boresight direction of a base station's antenna beam.}
  \label{fig:hex_cell}
\end{figure}

In the numerical experiments, $50$ instances of shadow fading realizations of the network are generated.
For each shadow fading realization, $10$ Raleigh fading instances are generated (i.e., there are a total of $500$ sets of channel realizations).
The convex optimization problems are solved using the primal-dual interior-point algorithm described in \cite[Section~11.7]{boyd04:convex_opt}.
Fig.~\ref{fig:avg_completion_time} shows the average completion time in different transmission schemes:
i) full power; ii) completion time minimization; and iii) robust power control
subject to different outage probability constraints $q=0.05,0.1,0.15,0.2$,
where we assume a common outage probability constraint (i.e., $q_i = q$, $i=1,\dotsc,57$).
Full power transmission, which is used as a baseline in the comparisons in this section, refers to the traditional scheme in wireless networks in which each base station transmits at its full power, and the transmit power is undifferentiated among the users: $P_i=\bar{P}_i$, for $i=1,\dotsc,M$.
Each mobile provides feedback on the realized SINR $S_i$ to its base station so that the encoding rate is set accordingly.
Completion time minimization refers to the solution of (\ref{eq:cvx_opt_min_J_R})--(\ref{eq:cvx_opt_yi_Pbi}), and
robust power control refers to the solution of (\ref{eq:Ray_min_J_R})--(\ref{eq:Ray_Pi_Pbi}).
In both cases, we minimize the sum completion time, i.e.,
we set the objective function $J(\mathbf{T}) = \sum_{i=1}^{57} T_i$
in (\ref{eq:cvx_opt_min_J_R}) and (\ref{eq:Ray_min_J_R}).

The cumulative distribution function (CDF) curves of the completion time in the different transmission schemes are exhibited in Fig.~\ref{fig:CDF_completion_time}.
In completion time minimization, we assume global channel knowledge, and an instance of the optimization problem (\ref{eq:cvx_opt_min_J_R})--(\ref{eq:cvx_opt_yi_Pbi}) is solved for each of the $500$ sets of random channel realizations.
On the other hand, under robust power control, the transmitters know only the shadow fading realizations but not the fast Raleigh fading realizations (the receivers have perfect channel knowledge).
Thus an instance of the optimization problem (\ref{eq:Ray_min_J_R})--(\ref{eq:Ray_Pi_Pbi}) is solved for each of the $50$ sets of shadow fading realizations, and the same solution (i.e., the transmission power $P_i^{\star}$ and target SINR $S_i^{\star}$) is used in each of the $10$ instances of fast Rayleigh fading associated with the shadow fading realization.
Shown in Figs.~\ref{fig:avg_completion_time} and~\ref{fig:CDF_completion_time} under robust power control are the completion times associated with the target SINR $S_i^{\star}$.
An outage event occurs if the channel realization cannot support the target SINR\@.
The empirical distribution of the number of users (out of $57$) in outage is plotted in Fig.~\ref{fig:outage}.

\begin{figure}
  \centering
  \includegraphics*[width=12cm]{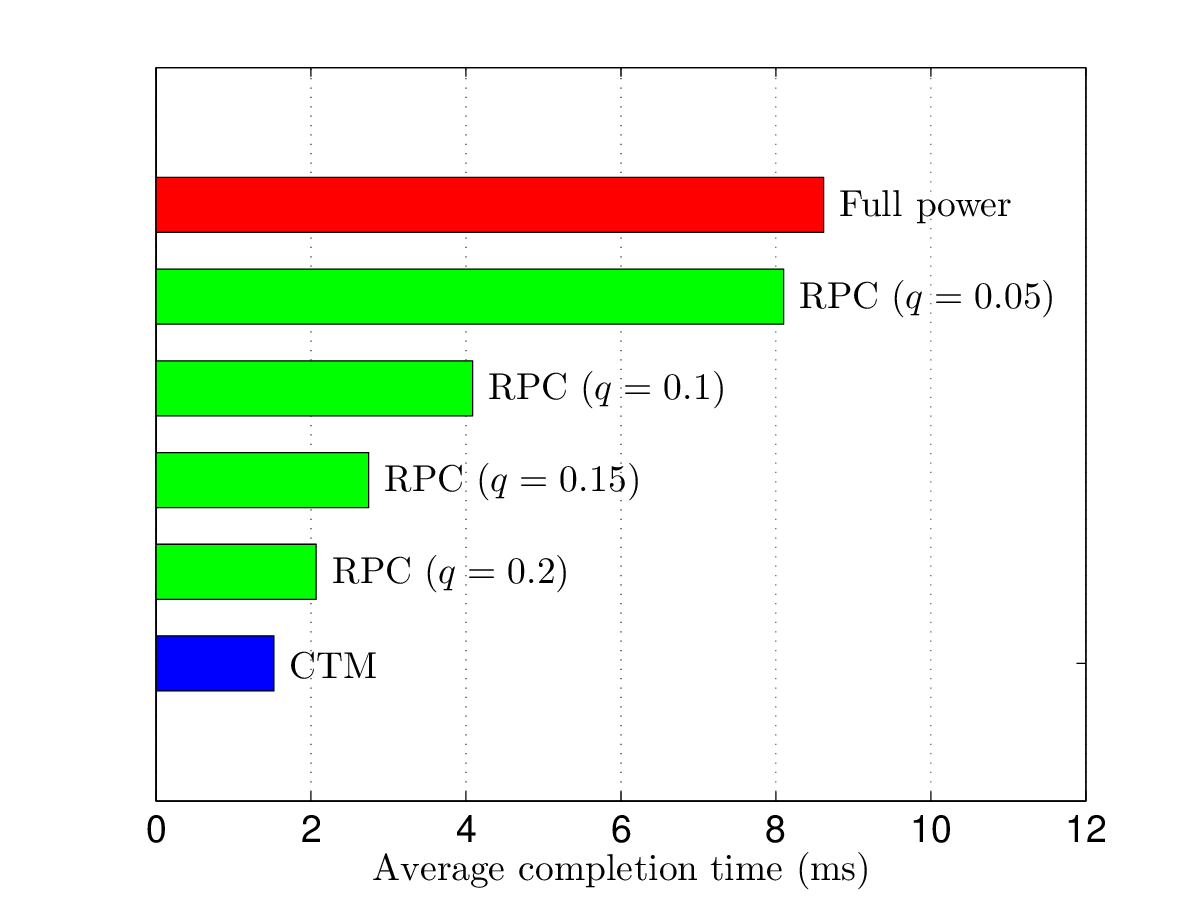}
  \caption{Average completion time in different transmission schemes: full power, completion time minimization (CTM), and robust power control (RPC) subject to different outage probability constraints.}
  \label{fig:avg_completion_time}
\end{figure}

\begin{figure}
  \centering
  \includegraphics*[width=12cm]{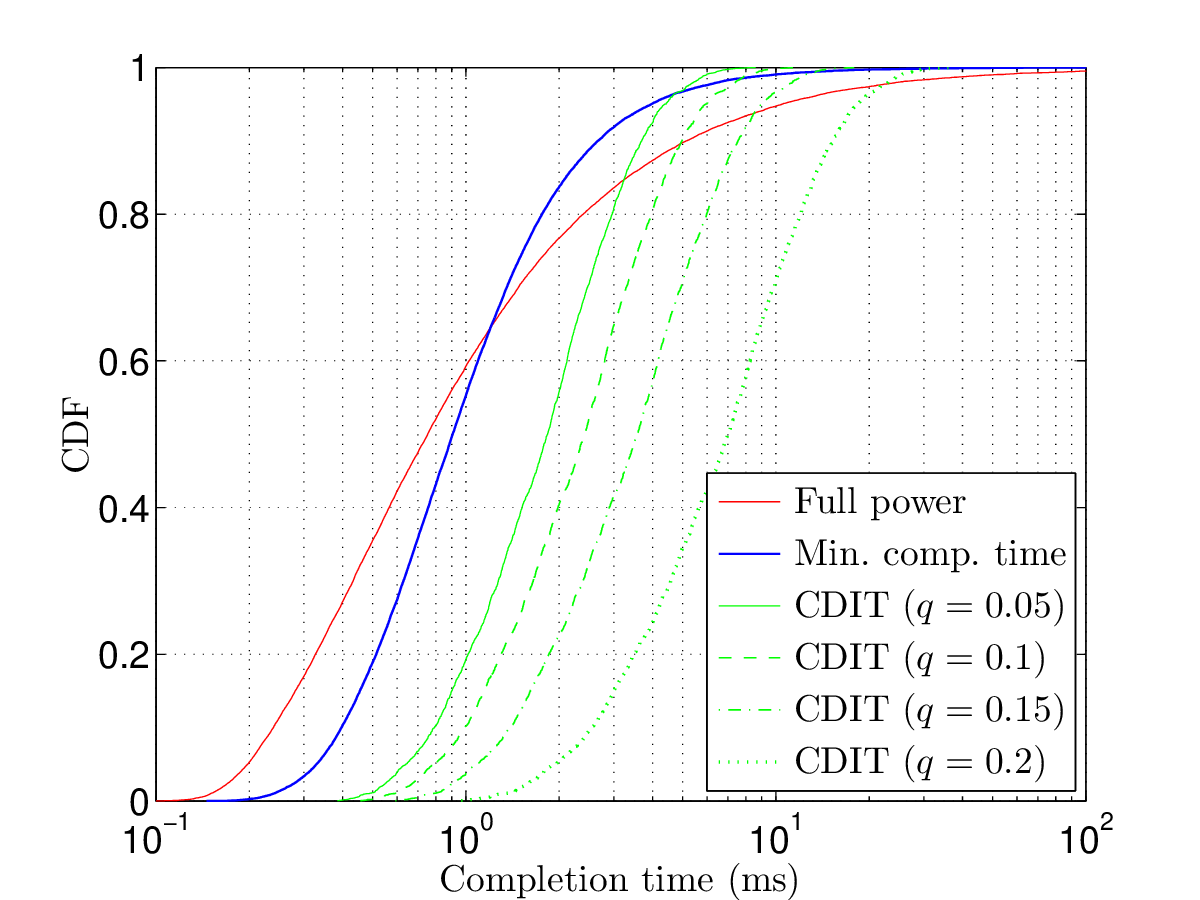}
  \caption{CDF curves of the completion time in different transmission schemes: full power, completion time minimization (CTM), and robust power control (RPC) subject to different outage probability constraints.}
  \label{fig:CDF_completion_time}
\end{figure}

\begin{figure}
  \centering
  \includegraphics*[width=12cm]{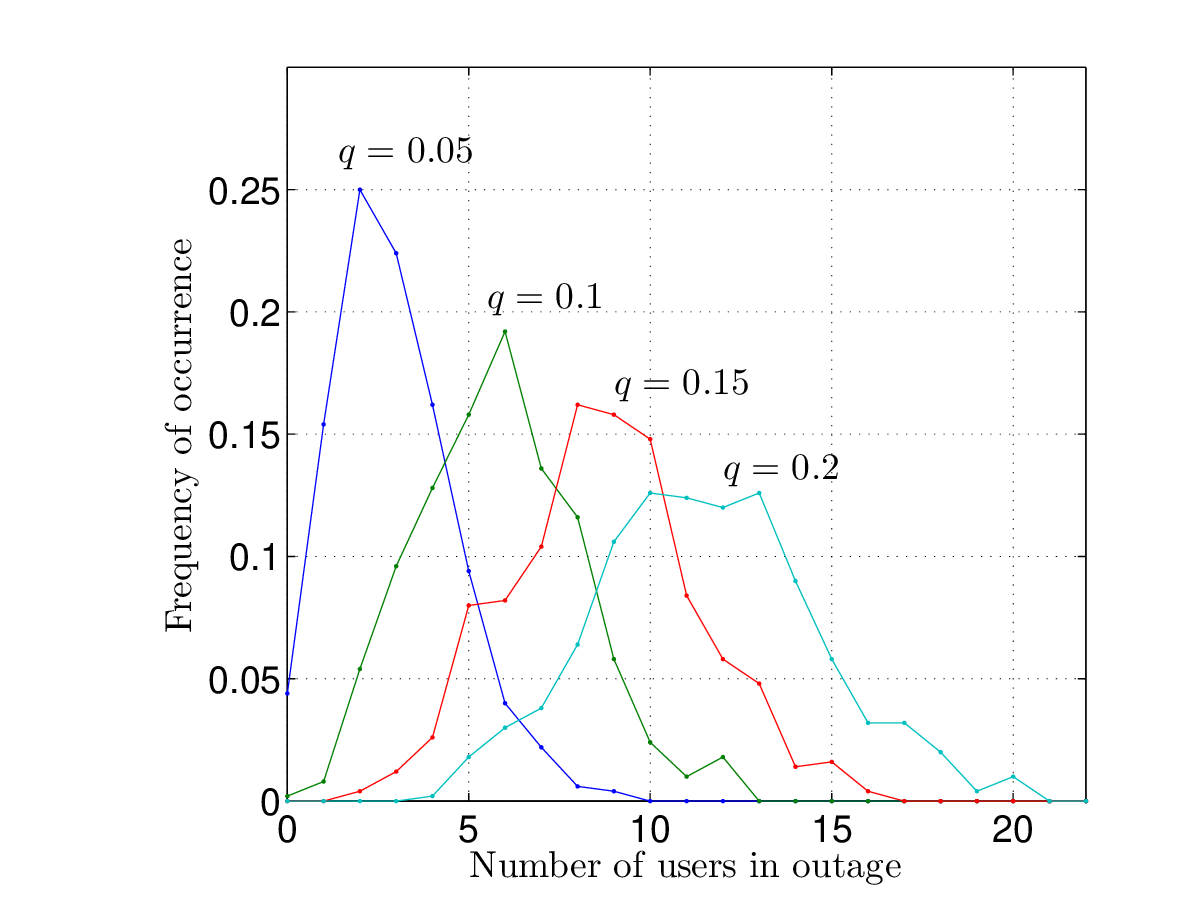}
  \caption{Empirical distribution of the number of users (out of $57$) in outage under robust power control. The outage probability constraints are labeled next to their corresponding curves.}
  \label{fig:outage}
\end{figure}

Fig.~\ref{fig:avg_completion_time} shows that under optimal power allocation, the average completion time is reduced approximately $82\%$ as compared to full power transmission.
Intuitively, the overall network performance is improved when, under each channel realization, the users with favorable channel conditions would power down their transmission to reduce interference to those with unfavorable channel conditions.
Fig.~\ref{fig:CDF_completion_time} shows that the completion time minimization scheme virtually eliminated excessively long completion times.
Under robust power control, when only channel distribution information is available at the transmitters, it results in longer completion times, and increasingly so with more stringent outage probability constraints.
In Fig.~\ref{fig:outage}, the empirical number of users in outage matches well with the expected number of users in outage ($57q$).

\section{Conclusions}
\label{sec:conclu}

In this paper, we consider minimizing a convex function of the completion times of user packets by optimally allocating transmission power in a wireless network.
We first focus on the scenario where the channel gains can be estimated accurately and are known by all users.
We show that completion time minimization can be formulated as a convex optimization problem, and hence the corresponding optimal power allocation can be efficiently computed.
The optimization formulation is valid for all ranges of SINR, which is especially pertinent for wireless sensor networks or delay-insensitive data applications where the users may transmit at moderate or low SINRs.
Under fading channels with transmission power adaptation across fading states, an average power constraint can be incorporated into the optimization problem.
We show that completion time minimization is a special case of rate utility maximization for which the optimization problem can be posed as convex.
In particular, in a wireless network, although the feasible rate region is non-convex, the corresponding completion time region is shown to be convex.
Furthermore, we consider robust power control under imperfect channel knowledge in fading channels.
Completion times are minimized subject to outage probability constraints over the fading distribution,
and we show that for a wide class of commonly used fading distributions, e.g., Rayleigh, Nakagami, and log-normal, robust power control can be posed as a convex optimization problem.
Finally, we apply the completion time minimization and robust power control frameworks in the setting of a wireless cellular network,
and show that optimizing the transmission power can significantly reduce the average completion time as compared to full power transmission.

\appendix

\subsection{Proof of Theorem~\ref{thm:rlb_fn_log_ccv}}
\label{sec:prf_thm_rlb_fn_log_ccv}

We first introduce, in terms of the transformed variables in (\ref{eq:t_W_ij_ln_W_ij}), the notation of
$\tilde{\mathbf{W}}_{-i} \in \mathds{R}^{M-1}$ representing the interfering channel random vector
\begin{align}
\tilde{\mathbf{W}}_{-i} &\triangleq \begin{bmatrix}\tilde{W}_{i1} \;\dots \;\tilde{W}_{i\,i-1} \;\tilde{W}_{i\,i+1} \;\dots \;\tilde{W}_{iM}\end{bmatrix}^T
\end{align}
and $\tilde{\mathbf{w}_{-i}}$ is a realization of $\tilde{\mathbf{W}}_{-i}$.
The PDF of $\tilde{\mathbf{W}}_{-i}$ is given by the marginal
\begin{align}
\label{eq:f_tW_mi_int_f_tW_i}
f_{\tilde{\mathbf{W}}_{-i}}(\tilde{\mathbf{w}}_{-i}) =
\int f_{\tilde{\mathbf{W}}_i}(\tilde{\mathbf{w}}_i) \;d\tilde{w}_{ii}.
\end{align}
Next, conditioning on $\tilde{\mathbf{W}}_{-i} = \tilde{\mathbf{w}}_{-i}$, the reliability function in (\ref{eq:Phi_i_Pr_W_S_P}) is given by
\begin{align}
\label{eq:Phi_i_int_phi_i_w}
\Phi_i(\tilde{S}_i, \tilde{\mathbf{P}}) =
\int \phi_i(\tilde{S}_i, \tilde{\mathbf{P}},\tilde{\mathbf{w}}_{-i}) \;\mathbf{d}\tilde{\mathbf{w}}_{-i}
\end{align}
where $\phi_i(\tilde{S}_i, \tilde{\mathbf{P}},\tilde{\mathbf{w}}_{-i})$ is defined as the composition of
\begin{align}
\label{eq:phi_i_bF_ftW}
\phi_i(\tilde{S}_i, \tilde{\mathbf{P}},\tilde{\mathbf{w}}_{-i}) &\triangleq
\bar{F}_{\tilde{W}_{ii}}\bigl(g_i(\tilde{S}_i, \tilde{\mathbf{P}}, \tilde{\mathbf{w}}_{-i})\bigr)\,
f_{\tilde{\mathbf{W}}_{-i}}(\tilde{\mathbf{w}}_{-i})\\
\label{eq:b_F_W_ii_1F}
\bar{F}_{\tilde{W}_{ii}}(\tilde{w}_{ii}) &\triangleq 1 - F_{\tilde{W}_{ii}}(\tilde{w}_{ii})\\
\label{eq:gi_Si_tPi_twmi}
g_i(\tilde{S}_i, \tilde{\mathbf{P}}, \tilde{\mathbf{w}}_{-i}) & \triangleq
\ln \biggl\{ N_i\exp(\tilde{S}_i-\tilde{P}_i)
+ \sum_{j\neq i} \exp(\tilde{S}_i-\tilde{P}_i+\tilde{P}_j+\tilde{w}_{ij})\biggr\}.
\end{align}
In (\ref{eq:b_F_W_ii_1F}), $F_{\tilde{W}_{ii}}(\tilde{w}_{ii})$ is the CDF of $\tilde{W}_{ii}$,
and $\bar{F}_{\tilde{W}_{ii}}(\tilde{w}_{ii})$ is referred to as its complementary CDF, which is a nonincreasing function in $\tilde{w}_{ii}$.

\begin{IEEEproof}
In the construction of $\Phi_i(\tilde{S}_i, \tilde{\mathbf{P}})$ in (\ref{eq:Phi_i_int_phi_i_w}),
with the application of Lemma~\ref{lem:f_tWi_log_ccv} below,
log-concavity is preserved \cite{prekopa95:stoc_prog, boyd04:convex_opt} under the integration in (\ref{eq:Phi_i_int_phi_i_w}), (\ref{eq:f_tW_mi_int_f_tW_i}); multiplication in (\ref{eq:phi_i_bF_ftW}); complementary CDF in (\ref{eq:b_F_W_ii_1F}); and composition of a logarithmically concave, nonincreasing function with a convex function in (\ref{eq:gi_Si_tPi_twmi}).
\end{IEEEproof}
\begin{lemma}
\label{lem:f_tWi_log_ccv}
Under transformation (\ref{eq:t_W_ij_ln_W_ij}), $f_{\tilde{\mathbf{W}}_i}(\tilde{\mathbf{w}}_i)$ is log-concave in $\tilde{\mathbf{w}}_i$.
\end{lemma}
\begin{IEEEproof}
Consider the logarithm of the PDF of transformed random vector $\tilde{\mathbf{W}}_i$
\begin{align}
\ln f_{\tilde{\mathbf{W}}_i}(\tilde{\mathbf{w}}_i) &=
\ln \biggl\{
\exp\Bigl(\sum_{j=1}^M \tilde{w}_{ij}\Bigr)
f_{\mathbf{W}_i}\bigl(\exp(\tilde{\mathbf{w}}_i)\bigr)
\biggr\}\\
&= \sum_{j=1}^M \tilde{w}_{ij} + \ln f_{\mathbf{W}_i}\bigl(\exp(\tilde{\mathbf{w}}_i)\bigr)
\end{align}
where log-concavity of $f_{\mathbf{W}_i}\bigl(\exp(\tilde{\mathbf{w}}_i)\bigr)$ in $\tilde{\mathbf{w}}_i$ follows from the condition given in Theorem~\ref{thm:rlb_fn_log_ccv}.
\end{IEEEproof}

\subsection{Proof of Proposition~\ref{prop:ray_naka_lognorm}}
\label{sec:prf_prop_ray_naka_lognorm}

When the channels exhibit independent fading, the joint distribution is given by the product of the marginal fading distributions.
Since log-concavity is preserved under multiplication \cite{prekopa95:stoc_prog, boyd04:convex_opt}, we show that each of the following marginal distribution satisfies the condition given in Theorem~\ref{thm:rlb_fn_log_ccv}.
\begin{IEEEproof}
The Rayleigh fading distribution is
\begin{align}
f_R(w) &= G^{-1} \exp(-w/G),\quad w\geq 0
\end{align}
where $G$ is the average channel power gain. The logarithm of the distribution in exponentiated variable is
\begin{align}
\ln f_R\bigl(\exp(\tilde{w})\bigr) &= -\ln G - \exp(\tilde{w})/G
\end{align}
which is a concave function of $\tilde{w}$.
Similarly, the Nakagami and log-normal fading distributions are respectively given by
\begin{align}
f_N(w) &= \frac{(m/G)^m w^{m-1} e^{-mw/G}}{\Gamma(m)}, \quad w\geq 0\\
f_L(w) &= \frac{10/\ln 10}{\sqrt{2\pi}\sigma w}\exp\biggl(-\frac{(10\log_{10}w-\mu)^2}{2\sigma^2}\biggr), \quad w\geq 0
\end{align}
where $\Gamma(\cdot)$ is the gamma function, and $m\geq1/2$,  $\mu,\sigma>0$ are the parameters of the fading distributions.
The logarithm of the distributions in exponentiated variables are
\begin{align}
\ln f_N\bigl(\exp(\tilde{w})\bigr) &= -\ln\Gamma(m) +m\ln(m/G) +(m-1)\tilde{w} -m \exp(\tilde{w})/G\\
\ln f_L\bigl(\exp(\tilde{w})\bigr) &=  \ln\biggl(\frac{10/\ln 10}{\sqrt{2\pi}\sigma}\biggr)
- \tilde{w} - \frac{(10\tilde{w}/\ln 10 - \mu)^2}{2\sigma^2}
\end{align}
which are concave functions.
\end{IEEEproof}

\section*{Acknowledgment}
The authors would like to thank Sivarama Venkatesan for providing the hexagonal cellular network channel modeling software.


\bibliographystyle{IEEEtran}
\bibliography{IEEEabrv,wrlscomm}


\end{document}